\documentclass[11pt, a4paper]{article} 
\usepackage[margin=1in]{geometry}

\usepackage{hyperref}

\usepackage{amssymb}
\usepackage{mathtools}
\usepackage{nicefrac}

\usepackage{amsthm}
\newtheorem{theorem}{Theorem}
\newtheorem{corollary}[theorem]{Corollary}
\newtheorem{lemma}[theorem]{Lemma}
\newtheorem{proposition}[theorem]{Proposition}
\theoremstyle{definition}
\newtheorem{example}[theorem]{Example}

\usepackage{subcaption}
\usepackage{xcolor}
\definecolor{color0}{HTML}{e6e6e6}
\definecolor{color1}{HTML}{2578C1}
\definecolor{color2}{HTML}{E2463E}
\definecolor{color3}{HTML}{B6DFEC}
\definecolor{color4}{HTML}{FDB872}

\usepackage{tikz}
\usetikzlibrary{patterns}
\tikzset{job1/.style={rectangle,draw,anchor=west,minimum height=0.5cm,minimum width=0.5cm, fill=color1}}
\tikzset{job2/.style={rectangle,draw,anchor=west,minimum height=0.5cm,minimum width=0.5cm, fill=color2}}				
\tikzset{job3/.style={rectangle,draw,anchor=west,minimum height=0.5cm,minimum width=0.5cm, fill=color3}}
\tikzset{job_l/.style={rectangle,draw,anchor=west, inner sep=0pt, minimum height= 0.5cm}}

\newcommand{\E}[0]{\mathbb{E}}
\newcommand{\Z}[0]{\mathbb{Z}}
\newcommand{\N}[0]{\mathbb{N}}

\usepackage[authoryear]{natbib}
 \bibpunct[, ]{(}{)}{,}{a}{}{,}%
\begin{document}

\title{Flow Shop Scheduling with Stochastic Reentry}

\author{Maximilian von Aspern$^1$ \and Felix Buld$^1$ \and Michael Pinedo$^2$
}

\maketitle
\vspace{-1.5em}
{\centering\small 
${}^1$ Chair of Operations Research, Technical University of Munich, Arcisstr. 21, 80333 Munich, Germany \\
${}^2$ Leonard N. Stern School of Business, New York University, 44 W 4th St., New York, NY, 10012, USA}

\abstract{We study flow shop scheduling with stochastic reentry, where jobs must complete multiple passes through the entire shop, and the number of passes that a job requires for completion is drawn from a discrete probability distribution.
The goal is to find policies that minimize performance measures in expectation.
Our main contribution is a reduction to a stochastic scheduling problem on identical parallel machines augmented by machine arrivals. This reduction preserves objective values and enables the transfer of structural results and performance guarantees from the auxiliary problems to the reentrant flow shop setting.
We demonstrate the usefulness of this reduction by proving the optimality of simple priority policies for minimizing the makespan and the total completion time in expectation under geometric and, more generally, monotone hazard rate distributions.
For minimizing the total weighted completion time, we derive an approximation guarantee for a simple priority policy that depends only on the squared coefficient of variation of the underlying distributions.
Our results constitute the first optimality and approximation guarantees for flow shops with stochastic reentry and demonstrate that established scheduling policies naturally extend to this setting through the proposed reduction.}

\section{Introduction}

The machine environment \emph{flow shop with reentry} is closely related to the well-understood flow shop and job shop environments.
In classical flow shops, each job must traverse each machine exactly once; in contrast, in this model, jobs may require several passes through the entire flow shop to be completed.
We call each pass through the flow shop a loop and allow jobs to require a different number of loops for completion.
After completing a pass through the flow shop, the job may either be scheduled again on the shop's first machine immediately or re-enter later.
Although general job shops can model this behavior, job shop scheduling policies may not exploit its inherent cyclical structure.
 
The reentry pattern can be used to model manufacturing processes in which multiple layers of material are applied individually, e.g., textile dyeing or mirror fabrication \citep*{ChoiKim2007}.
A prominent application arises in semiconductor manufacturing, where dozens of layers of insulators and conductors are individually deposited onto a wafer through a series of processes on specialized machines \citep*{pfund2006, Moench2011}.
Since wafers pass through the entire flow shop for each layer and require a different number of layers depending on the architecture or function of the desired semiconductor chip, this problem naturally motivates the study of the flow shop model with reentry.

Real-world manufacturing problems are commonly subject to uncertainty, for example, due to processing errors or malfunctions.
To account for such uncertainties, we consider a general model in which the numbers of required loops are randomly distributed and aim to develop policies that are optimal in expectation.
Motivated by takt times in industrial manufacturing settings, we assume that each individual operation requires unit time, independent of the job or loop it belongs to or the machine that processes it.

This paper extends the work of \cite{yu2020} and \cite*{aspern2025}, which consider the same problem in deterministic settings, where the numbers of required loops are known a priori.
We show that when the number of required loops for job~$j$ is drawn from a geometric distribution, or more generally, a monotone-hazard-rate distribution, there exist optimal policies that are natural extensions of optimal policies in deterministic settings.
Because the staggered processing structure prevents a direct application of classical inductive techniques from the stochastic scheduling literature used to derive optimality results, we propose a refined approach to obtain the results presented here. In particular, we construct a set of auxiliary instances of a related parallel-machine scheduling problem such that aggregating the job completion times across these instances exactly recovers the completion times in the corresponding flow shop schedule.
To the best of our knowledge, this is the first study of stochasticity in flow shops with reentry that establishes optimality or approximation guarantees.

\subsection{Model and Definitions}
Throughout this paper, we use the shorthand notation $[N] \coloneqq \{1,2,\ldots,N\}$ to denote the set of positive integers up to $N$.
We study a flow shop consisting of $m$ machines $[m]$ that must be traversed in increasing order and multiple times by $n$ jobs.
For a job~$j \in [n]$, the discrete and integrable random variable $Y_j$ denotes the number of required loops through the flow shop for completion. Loop $l \in [Y_j - 1]$ of job $j$ must precede loop $l+1$, i.e., it must be completed on machine~$m$ before loop $l+1$ may be started on machine $1$.
After job $j \in [n]$ completes its $l$-th loop, the scheduler learns whether $Y_j = l$ or $Y_j > l$.
Furthermore, each machine can process at most one job at a time.
The processing of any job and any loop on any machine takes unit time (i.e., $p_{ijl} = 1$ for all $i \in [m], j \in [n], l \in [Y_j]$ in standard scheduling notation). We assume that $Y_1, \ldots, Y_n$ are mutually independent.
When minimizing a weighted objective function, each job~$j \in [n]$ is also assigned a positive weight $w_j$.
For a given objective function~$\gamma$, the goal is to compute a scheduling policy that minimizes $\gamma$ in expectation.
In this paper, we focus on the objectives of minimizing the makespan, the total completion time, and the total weighted completion time, i.e., $\gamma\in \lbrace C_{\max}, \sum C_j, \sum w_jC_j\rbrace$.

We refer to this problem as \textsc{Flow Shop Scheduling with Stochastic Reentry}.
In the standard three-field-notation, it may be denoted as $F |rntr, p_{ijl}=1,  Y_j \sim \textup{stoch} | E\left[ \gamma \right] $.

Because all processing times are unit, each loop can, without loss of generality, be scheduled contiguously, and we can restrict our attention to schedules in which all operations start at integer times. See Figure~\ref{fig:example} for an illustration of the problem's structure.

The scheduling policies we consider in this paper are dynamic and non-anticipatory.
Dynamic, non-anticipatory policies may use already realized information to make decisions, that is, the number of already completed loops for each job and whether it has already been completed. But their decisions may not depend on unobserved information about the realizations of uncompleted jobs.
In particular, a policy checks at each time step if a job has completed a loop through the flow shop without being fully processed.
If so, this job remains in the system and becomes available for processing again. A job is called available for processing at a given time if it is not yet completed and is not currently being processed by any machine.
Then, the policy either selects an available job to process on machine~$1$ or idles machine~$1$.
It has no knowledge of the exact number of loops a job requires, but only of its distribution.

\begin{figure}[b]
    \centering
    \begin{tikzpicture}[thick]
			\draw[->] (0,0) -- (5,0) node[below]{$t$};
            \draw[-] (0,0) -- (0,2) node[below]{};
            
            \draw (0,0) -- (0,-0.2) node[below]{$0$};
            \draw (1.5,0) -- (1.5,-0.2) node[below]{$3$};
            \draw (3,0) -- (3,-0.2) node[below]{$6$};
            \draw (4.5,0) -- (4.5,-0.2) node[below]{$9$};

            {\draw[fill=color1]  (0,1) rectangle (0.5,1.5) node[pos=.5] {$1$};}
            {\draw[fill=color1] 	 (0.5,0.5) rectangle (1,1) node[pos=.5] {$1$};}
            {\draw[fill=color1]	 (1,0) rectangle (1.5,0.5) node[pos=.5] {$1$};}

            {\draw[fill=color1]	 (2,1) rectangle (2.5,1.5) node[pos=.5] {$1$};}
            {\draw[fill=color1] 	 (2.5,0.5) rectangle (3,1) node[pos=.5] {$1$};}
            {\draw[fill=color1] 	 (3,0) rectangle (3.5,0.5) node[pos=.5] {$1$};}

            {\draw[fill = color2] 	 	 (1.5,1) rectangle (2,1.5) node[pos=.5] {$2$};}
            {\draw[fill = color2] 		 (2,0.5) rectangle (2.5,1) node[pos=.5] {$2$};}
            {\draw[fill = color2] 	 	 (2.5,0) rectangle (3,0.5) node[pos=.5] {$2$};}

            {\draw[fill = color2]	 (3,1) rectangle (3.5,1.5) node[pos=.5] {$2$};}
            {\draw[fill = color2]	 (3.5,0.5) rectangle (4,1) node[pos=.5] {$2$};}
           {\draw[fill = color2]	 (4,0) rectangle (4.5,0.5) node[pos=.5] {$2$};}

            {\draw[fill = color3] 	 (1,1) rectangle (1.5,1.5) node[pos=.5] {$3$};}
            {\draw[fill = color3]  	 (1.5,0.5) rectangle (2,1) node[pos=.5] {$3$};}
            {\draw[fill = color3]  	 (2,0) rectangle (2.5,0.5) node[pos=.5] {$3$};}
            
            {\draw[fill = color4] 	 (0.5,1) rectangle (1,1.5) node[pos=.5] {$4$};}
            {\draw[fill = color4] 	 	 (1,0.5) rectangle (1.5,1) node[pos=.5] {$4$};}
            {\draw[fill = color4] 	 	 (1.5,0) rectangle (2,0.5) node[pos=.5] {$4$};}
		\end{tikzpicture}
    \caption{A feasible schedule for $m=3$ machines and a realization in which $Y_1 = Y_2 = 2, Y_3 = Y_4 = 1$. Job~$1$ is interrupted after its first loop, whereas the loops of job~$2$ are processed consecutively. Contiguous processing of operations within each loop gives rise to the characteristic staircase shape.}
    \label{fig:example}
\end{figure}

More concisely, a \emph{policy} $\pi$ is a function that takes as input an instance~$I$, the current time~$t \in \mathbb{N}_0$, all previous scheduling decisions, and, for each job the information, whether it has already been completed,
and outputs an available job~$j$ or a unit-length idle period to be scheduled next on the first machine.
In general, such a policy may require exponential space to encode or exponential time to evaluate due to the size of the state space.
We are particularly interested in simple static priority policies that assign a fixed priority value to each job and always schedule the job with the highest priority among all available jobs next.

Let OPT denote an optimal policy. We call $\pi$ an $\alpha$-\emph{approximation} 
if $\E[\gamma(\pi(I))] \leq \alpha \E[\gamma(\textup{OPT(I)})]$ for all instances $I$, where $\E[\gamma(\pi(I))]$ denotes the expected objective value of policy $\pi$ on instance~$I$. We call $\alpha$ the \emph{approximation ratio} or \emph{performance guarantee} under policy $\pi$.

\subsection{Related Work}
Flow shops with reentry were originally introduced by \cite*{Graves1983}, and various heuristics have since been developed \cite*[e.g., by][]{Mason2002,pan2003,Choi2008,Jing2011,segal2026}. Polynomial-time solvable special cases or approximation guarantees have so far mostly been analyzed for instances with unit processing times in environments with reentry, for example, by \cite*{Brucker1999,Timkovsky2003,yu2020,shufan2023}, and \cite{aspern2025}.
Other reentrant flow shop models, in which jobs re-enter only some of the flow shop's machines, were also studied. For an overview, see, for example, \cite{Middendorf2002,danping2011}, and \cite{Emmons2012}.

The literature on stochasticity in flow shops with reentry remains sparse. The few existing works in this area mostly focus on deterministic reentry and stochastic processing times \cite*[e.g.,][]{dugardin2010,liu2024}, or on application-oriented reactive rescheduling approaches based on job rework \cite*{lee2011}. 

By contrast, there is a large body of results on stochastic scheduling in traditional flow shops and on parallel machines; see \cite{Emmons2012} and \cite{pinedo2022} for comprehensive overviews.
The optimality of simple priority policies for the makespan and total completion time objectives when scheduling jobs with geometrically (or exponentially) distributed processing times on parallel machines was established by various authors \citep*{Pinedo1979,Bruno1981,Gittins1981,Heyden1981,Weber1982,Kampke1987,Chang1991}.
Most relevant to the problem described in this paper are the results by \cite{glazebrook1979} and \cite{Weber1979} on parallel machine models, in which the number of machines increases monotonically over time.
The authors showed that simple priority policies are optimal for a broad class of random processing time distributions, including geometric and deterministic processing times.
For the total weighted completion time objective, the first policy with a parameterized approximation ratio for arbitrary distributions was given by \cite*{Mohring1999}.
Its approximation ratio was later improved by \cite{jager2018}.

\subsection{Outline and Results}

In Section~\ref{sec:reduction}, we present our reduction to an auxiliary parallel machine scheduling problem.
Specifically, we provide sufficient conditions for the existence of policies with bounded performance guarantees in the reentrant flow shop setting.
Section~\ref{sec:LERL&MERL} contains several applications of our reduction.
We state simple priority policies for the makespan and the total completion time objective and show that they are optimal under certain assumptions.
For the total weighted completion time objective, we show that an equally simple priority policy achieves a parameterized approximation ratio. Technical details of the analysis are contained in Section~\ref{sec:WSEPT}. We extend our results to practically motivated distributions in Section~\ref{sec:Extensions} and conclude with a discussion in Section~\ref{sec:outlook}.

\section{Reduction to Stochastic Scheduling with Machine Arrivals}
\label{sec:reduction}

In this section, we present our main technical contribution: a reduction to an auxiliary problem, which we shall refer to as \textsc{Stochastic Scheduling with Machine Arrivals}.
This is a classical scheduling problem on $m$ identical parallel machines augmented with machine arrivals.
A machine~$i \in [m]$ arrives at time $a_i \in \Z$ and must not be used before this time.
Once a machine arrives, it remains available until all jobs have been completed. 
A random variable $X_j$ denotes the processing time of job~$j$, and its realization becomes known only upon completion of the job.
Jobs may be preempted, i.e., their processing can be halted and resumed later on any machine.
Using the standard three-field notation, one may denote this problem as $P | a_i, X_j \sim \textup{stoch}, \textup{pmtn} | \E\left[ \gamma \right]$.
The reduction conceptually consists of two separate steps, as outlined below.
\begin{enumerate}
    \item Each instance of \textsc{Flow Shop Scheduling with Stochastic Reentry} can be viewed as an instance of \textsc{Stochastic Scheduling with Machine Arrivals} with restrictions on preemption and starting times: Loops (which each require $m$ time units) are mapped to processing blocks of length $m$, which are to be processed on $m$ identical parallel machines.
    Because loops are scheduled contiguously, preemption only occurs at the end of such a block.
    To capture that the loops are completed sequentially in the flow shop model, we let the machines become available at staggered times $0,1, \ldots, m-1$, and assign the block corresponding to a loop starting at time $\tau$ to machine $(\tau \mod m) + 1$.
    Figure~\ref{fig:reduction1} illustrates this first reduction step.
    \item Instead of working with this single instance of $P | a_i, X_j \sim \textup{stoch}, \textup{pmtn} | \E\left[ \gamma \right]$, we generate~$m$ auxiliary instances, in which each machine arrives either at time $0$ or at time~$1$. 
    Any feasible policy for \textsc{Flow Shop Scheduling with Stochastic Reentry} induces a feasible policy for each auxiliary instance, and the average behavior on these instances exactly matches the behavior on the original instance, as illustrated in Figure~\ref{fig:reduction2}.
\end{enumerate}

\begin{figure}
    \centering
    \begin{tikzpicture}[thick]
			\draw[->] (0,0) -- (5,0) node[below]{$t$};
            \draw[-] (0,0) -- (0,2) node[below]{};
            
            \draw (0,0) -- (0,-0.2) node[below]{$0$};
            \draw (1.5,0) -- (1.5,-0.2) node[below]{$3$};
            \draw (3,0) -- (3,-0.2) node[below]{$6$};
            \draw (4.5,0) -- (4.5,-0.2) node[below]{$9$};

            \draw[fill=color1]  (0,1) rectangle (1.5,1.5) node[pos=.5] {$1$};
        
            \draw[fill=color1] 	 	 (2,0.5) rectangle (3.5,1) node[pos=.5] {$1$};
            
            
            \draw[fill = color2]  	 (1.5,1) rectangle (3,1.5) node[pos=.5] {$2$};

            \draw[fill = color2]  	 (3,1) rectangle (4.5,1.5) node[pos=.5] {$2$};

            \draw[fill = color3] 	 	 (1,0) rectangle (2.5,0.5) node[pos=.5] {$3$};
            
            \draw[fill = color4] 	 (0.5,0.5) rectangle (2,1) node[pos=.5] {$4$};      

            \fill[pattern=north east lines] (0,0.5) rectangle (0.5,1);
            \fill[pattern=north east lines] (0,0) rectangle (1,0.5);

	\end{tikzpicture}
    \caption{The flow shop schedule from Figure~\ref{fig:example} transformed into a schedule on $m=3$ parallel machines with arrival times $0,1,2$. Each loop corresponds to a contiguous processing block of length $m = 3$, and completion times are preserved.}
    \label{fig:reduction1}
\end{figure}

\begin{figure}
    \centering
    \begin{subfigure}[b]{0.32\columnwidth}
    \centering
        \begin{tikzpicture}[thick]
			\draw[->] (0,0) -- (2,0) node[below]{$t$};
            \draw[-] (0,0) -- (0,2) node[below]{};
            
            \draw (0,0) -- (0,-0.2) node[below]{$0$};
            \draw (0.5,0) -- (0.5,-0.2) node[below]{$1$};
            \draw (1,0) -- (1,-0.2) node[below]{$2$};
            \draw (1.5,0) -- (1.5,-0.2) node[below]{$3$};

            \draw[fill=color1]  (0,1) rectangle (0.5,1.5) node[pos=.5] {$1$};
            \draw[fill= color1] 	 	 (0.5,0.5) rectangle (1,1) node[pos=.5] {$1$};
            
            \draw[fill = color2]  	 (0.5,1) rectangle (1,1.5) node[pos=.5] {$2$};
            \draw[fill = color2]  	 (1,1) rectangle (1.5,1.5) node[pos=.5] {$2$};
            
            \draw[fill = color3] 	 	 (0,0) rectangle (0.5,0.5) node[pos=.5] {$3$};
            
            \draw[fill = color4] 	 (0,0.5) rectangle (0.5,1) node[pos=.5] {$4$};
	    \end{tikzpicture}
        \caption{$A_1$}
    \end{subfigure}
    \begin{subfigure}[b]{0.32\columnwidth}
    \centering
        \begin{tikzpicture}[thick]
			\draw[->] (0,0) -- (2,0) node[below]{$t$};
            \draw[-] (0,0) -- (0,2) node[below]{};
            
            \draw (0,0) -- (0,-0.2) node[below]{$0$};
            \draw (0.5,0) -- (0.5,-0.2) node[below]{$1$};
            \draw (1,0) -- (1,-0.2) node[below]{$2$};
            \draw (1.5,0) -- (1.5,-0.2) node[below]{$3$};

            \draw[fill=color1]  (0,1) rectangle (0.5,1.5) node[pos=.5] {$1$};
            \draw[fill= color1] 	 	 (0.5,0.5) rectangle (1,1) node[pos=.5] {$1$};
            
            \draw[fill = color2]  	 (0.5,1) rectangle (1,1.5) node[pos=.5] {$2$};
            \draw[fill = color2]  	 (1,1) rectangle (1.5,1.5) node[pos=.5] {$2$};

            \draw[fill = color3] 	 	 (0.5,0) rectangle (1,0.5) node[pos=.5] {$3$};
            
            \draw[fill = color4] 	 (0,0.5) rectangle (0.5,1) node[pos=.5] {$4$};
           
            \fill[pattern=north east lines] (0,0) rectangle (0.5,0.5);
    	\end{tikzpicture}
        \caption{$A_2$}
    \end{subfigure}
    \begin{subfigure}[b]{0.32\columnwidth}
    \centering
        \begin{tikzpicture}[thick]
			\draw[->] (0,0) -- (2,0) node[below]{$t$};
            \draw[-] (0,0) -- (0,2) node[below]{};
            
            \draw (0,0) -- (0,-0.2) node[below]{$0$};
            \draw (0.5,0) -- (0.5,-0.2) node[below]{$1$};
            \draw (1,0) -- (1,-0.2) node[below]{$2$};
            \draw (1.5,0) -- (1.5,-0.2) node[below]{$3$};

            \draw[fill=color1]  (0,1) rectangle (0.5,1.5) node[pos=.5] {$1$};
            \draw[fill= color1] 	 	 (1,0.5) rectangle (1.5,1) node[pos=.5] {$1$};

            \draw[fill = color2]  	 (0.5,1) rectangle (1,1.5) node[pos=.5] {$2$};
            \draw[fill = color2]  	 (1,1) rectangle (1.5,1.5) node[pos=.5] {$2$};

            \draw[fill = color3] 	 	 (0.5,0) rectangle (1,0.5) node[pos=.5] {$3$};
            
            \draw[fill = color4] 	 (0.5,0.5) rectangle (1,1) node[pos=.5] {$4$};
           
            \fill[pattern=north east lines] (0,0) rectangle (0.5,1);
	    \end{tikzpicture}
        \caption{$A_3$}
    \end{subfigure}
    \caption{The $m=3$ auxiliary instances $A_1, A_2, A_3$ with zero, one, and two machines arriving at time~$1$, respectively. Completion times aggregated over auxiliary instances equal completion times in Figure~\ref{fig:reduction1}.}
    \label{fig:reduction2}
\end{figure}

With this reduction, we exploit the problem's structure and enable a transfer of known techniques and results from parallel machine scheduling to \textsc{Flow Shop Scheduling with Stochastic Reentry}.
The first reduction step is based on ideas presented in \cite{Timkovsky2003} and \cite{aspern2025} and can be understood as a relabeling of the original problem (assuming we also allow idle time only in blocks of size $m$).
Because the processing blocks are staggered and may not be preempted, the resulting problem still prevents the direct use of inductive arguments based on simultaneous preemption, such as those by \cite{glazebrook1979} and \cite{Weber1979}.
The second step removes this obstacle by introducing instances in which all jobs can be preempted simultaneously at integer times. We show below that approximation guarantees can be transferred from the auxiliary instances to the reentrant flow shop problem in an aggregated manner.
This enables us to extend known strategies for parallel machine scheduling to these instances in Sections~\ref{sec:LERL&MERL} and~\ref{sec:WSEPT}.

Separating the reduction into two steps provides intuition and relates our construction to similar ideas that have previously appeared in the literature. In the formal description below, we skip the first intermediate step and instead construct the $m$ auxiliary instances directly from the original flow shop problem.

For an instance~$I$ of the problem \textsc{Flow Shop Scheduling with Stochastic Reentry} with $n$ jobs and $m$ machines, we define a set of instances $\{A_1(I), A_2(I), \ldots, A_m(I)\}$ of the problem \textsc{Stochastic Scheduling with Machine Arrivals} as follows. 
For each job~$j \in [n]$ with $Y_j$ loops in $I$, the instances~$A_k(I)$, $k\in[m]$, each contain a job of stochastic processing time $X_j = Y_j$.
For all instances $I$, the arrival time of machine~$i \in [m]$ in instance~$A_k(I)$ is given by:
\begin{equation*}
       a^k_i = 
       \begin{cases}
            0, & \text{ if } i + k \leq m + 1, \\
            1, & \text{ else.}
        \end{cases}
   \end{equation*}
In particular, there are only $m-k+1$ machines available at time $0$ in instance~$A_k(I)$, and the remaining $k-1$ machines arrive at time $1$ (as illustrated in Figure~\ref{fig:reduction2}). Key ideas behind this decomposition in auxiliary instances are that processing times of jobs are split equally among the auxiliary instances, and that the machine arrivals aggregate over all auxiliary instances to the machine arrivals in the problem in the first step, i.e., we have $\sum_{k=1}^m a_i^k = i-1$ for all $i\in[m]$. If the instance $I$ is clear from the context, we occasionally write $A_k$ instead of~$A_k(I)$ for brevity.

In the following, we show how any policy~$\pi$ for instance~$I$ of \textsc{Flow Shop Scheduling with Stochastic Reentry} can be used to construct policies for instance~$A_k(I)$ in a way that preserves the expected objective value.
Let $\tau \in \mathbb{N}_0$ be the time at which $\pi$ schedules a loop of a job~$j$ (or a unit of idle period) on machine~$1$.
We define $t \coloneqq \lfloor \nicefrac{\tau}{m}\rfloor \in \mathbb{N}_0$ and $i \coloneqq \left( \tau \mod m \right) + 1 \in [m]$, i.e., $\tau = m \cdot t + i - 1$.
Now we can construct a policy~$\pi_k$ for instance~$A_k(I)$ which schedules job~$j$ at time~$t + a^k_i$ on machine~$i$ for one time unit as illustrated in Figure~\ref{fig:reduction2}.
We call $\pi_k$ the policy induced by $\pi$ on instance~$A_k(I)$.

The following ``monotonicity-preserving'' properties of this time mapping are used repeatedly in the remainder of this section.

\begin{lemma}
    \label{lemma: time mapping}
    Let $ \tau_1 = m \cdot t_1 + i_1 -1$ and $ \tau_2 = m \cdot t_2 + i_2-1$ be two integral points in time with $i_1, i_2 \in \left[m\right]$ as above and $k\in\left[m\right]$. Then, $\tau_1\leq\tau_2$ implies $t_1 + a_{i_1}^k\leq t_2 + a_{i_2}^k$. If moreover $\tau_1 + m \leq \tau_2$, then $t_1 + a_{i_1}^k < t_2 + a_{i_2}^k$.
\end{lemma}
\begin{proof}
    Let $\tau_1\leq\tau_2$ and consider two cases. If $t_1+1\leq t_2$, then we have $t_1 + a_{i_1}^k\leq t_2 + a_{i_2}^k$, since $a_{i_1}^k,a_{i_2}^k \in [0,1]$. Otherwise, it must hold that $t_1 = t_2$ and $i_1 \leq i_2$. The inequality follows since $a_{i_1}^k \leq a_{i_2}^k$ by definition of $a_i^k$. For $\tau_1+m\leq\tau_2$, applying the same arguments with cases $t_1+1< t_2$ and $t_1+1= t_2$ completes the proof. 
\end{proof}

The following theorem shows that our reduction preserves completion times exactly (aggregated over the set of auxiliary instances), a crucial ingredient for establishing approximation guarantees and optimality results using our framework. Note that this theorem also applies to policies that use interruptions between two loops of a job (in this case, the induced policies will produce a preemption).

\begin{theorem}
    \label{theorem:policy} 
    Given an instance~$I$ and a policy $\pi$, the induced policy~$\pi_k$ is feasible for instance $A_k(I)$ and satisfies:
    \begin{equation*}
          C_j^\pi(y) = \sum_{k=1}^{m} C_j^{\pi_k}(y)
    \end{equation*}
    for all $k \in [m]$, all $j \in [n]$, and all realizations $y$ of processing times drawn from the distribution of $\prod_{j = 1}^n Y_j$.
\end{theorem}
\begin{proof}
    We first show that $\pi_k$ is feasible, i.e., that a machine processes at most one job at a time and a job occupies only one machine at a time.
    The former property holds by construction, and we prove the latter in the following.
    Let two loops~$l_1$ and~$l_2$ of job~$j$ be started at times $\tau_1$ and $\tau_2$ by $\pi$. Because $\pi$ schedules those loops feasibly, we have $\tau_1+m\leq\tau_2$. Thus, Lemma~\ref{lemma: time mapping} implies that those two loops are not scheduled concurrently by $\pi_k$. 
    
    Additionally, we must show that $\pi_k$ is non-anticipatory. 
    Let $l_2$ be a unit of job~$j_2$ to be scheduled by $\pi_k$.
    Because its scheduling decision is, by definition, dependent on the scheduling decision of $\pi$ for the corresponding loop, $\pi_k$ must use the same information as $\pi$ for this decision.
    At the time of scheduling the corresponding loop, the policy $\pi$ may use information realized at the completion of all previously completed loops.
    Let $l_1$ denote such a loop and $j_1$ its corresponding job.
    Now it suffices to show that unit $l_1$ of job~$j_1$ does not complete after the time of scheduling unit $l_2$ of $j_2$ under $\pi_k$. 
    That is, all stochastic information accessed by $\pi_k$ has already been realized in instance $A_k(I)$ at the time of accessing it.
    With the notation of $\tau_1$ for the completion time of $l_1$ and~$\tau_2$ for the starting time of $l_2$, and $\tau_1\leq \tau_2$, this follows from Lemma~\ref{lemma: time mapping}. 
     
    Due to the construction of the policies for the auxiliary instances, job~$j$ in $A_k(I)$ is processed for exactly one time unit for each loop that job~$j$ completes in $I$, and it holds that for any~$y$:
    \begin{equation*}
        C^{\pi_k}_j(y) = \frac{C^\pi_j(y) - (i-1)}{m} + a_i^k.
    \end{equation*}
    Here, $i\in[m]$ denotes the machine index on which the last unit of job~$j$ is scheduled by policy $\pi_k$ in realization~$y$ as defined above. Summing over all $k \in [m]$ and observing that $\sum_{k=1}^m a_i^k = i-1$ completes the proof.
\end{proof}

    While any non-anticipatory and feasible policy $\pi$ for $I$ induces non-anticipatory and feasible policies $\pi_k$ with the desired properties for each $A_k(I)$, $k \in [m]$, the reverse does not hold.
    In particular, there are feasible and non-anticipatory policies $\pi_k$ for $A_k(I)$ that cannot be induced by a policy $\pi$ for $I$. For example, consider a policy $\pi_k$ that allows preemptions at non-integer times.

We can now state our main theorem, which enables a black-box transfer of approximation guarantees from the auxiliary problem to the original problem for the objective functions $\gamma \in \{\sum C_j, \sum w_jC_j, C_{\max}\}$.

\begin{theorem}
\label{thm:approximation} 
    If $E[\gamma(\pi_k)] \leq \alpha \E[\gamma(\textup{OPT}(A_k(I)))]$ for all instances $I$ and all $k \in [m]$, then policy~$\pi$ is an $\alpha$-approximation for \textsc{Flow Shop Scheduling with Stochastic Reentry}.
\end{theorem}
\begin{proof}
    Let $\pi^*$ denote the optimal policy for the flow shop instance and $\pi^*_k$ its induced policies for all $k \in [m]$.
    Note that it need not hold that $\pi^*_k = \textup{OPT}(A_k(I))$.
    However, by optimality of $\textup{OPT}(A_k(I))$, and the assumption, we have for all $k \in [m]$:
    \begin{equation*}
    E[\gamma(\pi_k)] \leq \alpha \E[\gamma(\textup{OPT}(A_k(I)))] \leq \alpha \E[\gamma(\pi^*_k)].
    \end{equation*}
    
    Let us first consider the objective $\sum w_j C_j$. 
    Applying Theorem~\ref{theorem:policy} and summing over all jobs, we get for any realization $y$:
    \begin{equation*}
          \sum_j w_j C_j^{\pi}(y) = \sum_{k=1}^{m}\sum_j w_j C_j^{\pi_k}(y).
    \end{equation*}   
    Note that this equality also holds for any other feasible policy, in particular, the optimal policy~$\pi^*$.
    Taking expectation over $y$ and applying the assumption, we obtain:
    \begin{equation*}
    \begin{split}
        \E\left[\sum_j w_jC^\pi_j\right] = \sum_{k=1}^{m} \E\left[\sum_j w_j C^{\pi_k}_j\right]
        \leq  \sum_{k=1}^{m} \alpha \cdot \E\left[\sum_j w_jC^{\pi^*_k}_j\right] =\alpha \cdot \E\left[\sum_j w_j C^{\pi^*}_j \right].
    \end{split}
    \end{equation*}
    
    The result for the objective $\sum C_j$ follows immediately. 
    For the makespan objective, let $C^{\pi}_{\max}(y)$ denote the makespan under policy $\pi$ and realization $y$, and let $j'$ denote the job that attains the makespan. By the monotonicity argument from Lemma~\ref{lemma: time mapping}, the makespan under policy $\pi_k$ on instance $A_k(I)$ is also attained by job~$j'$ (and potentially other jobs simultaneously). By Theorem~\ref{theorem:policy}, we obtain:
    \begin{equation*}
        C_{\max}^{\pi}(y) = C_{j'}^{\pi}(y) = \sum_{k=1}^{m} C_{j'}^{\pi_k}(y) = \sum_{k=1}^{m} C_{\max}^{\pi_k}(y).
    \end{equation*}
    Again, this equality also holds for the optimal policy $\pi^*$. The proof now follows analogously. 
\end{proof}
This theorem can, in particular, be used for $\alpha=1$ to establish the optimality of a policy~$\pi$, as seen in Corollary~\ref{corollary:LERL}.

\section{Optimal and Approximately Optimal Priority Policies}
\label{sec:LERL&MERL}

In this section, we apply Theorem~\ref{thm:approximation} to derive optimality results or approximation guarantees for certain policies for \textsc{Flow Shop Scheduling with Stochastic Reentry} by leveraging results for \textsc{Stochastic Scheduling with Machine Arrivals}.
Specifically, we consider a class of priority policies that are of particular interest due to their simplicity. Given a predefined static priority value for each job $j$, \emph{static priority policies} schedule at any time the job~$j$ with the largest priority value among all available jobs. In particular, static priority policies for \textsc{Flow Shop Scheduling with Stochastic Reentry} are \emph{non-interruptive}, i.e., whenever a job re-enters the system after being processed on machine~$m$, it is immediately scheduled on machine~$1$. 

\begin{lemma}
    \label{lem:priority}
    If $\pi$ is a static priority policy,
    then the induced policy $\pi_k$ is also a static priority policy, and it uses the same priority values as $\pi$ for all $k \in [m]$.
\end{lemma}
\begin{proof}
    First, observe that for all $k\in\left[m\right]$, $\pi_k$ is non-preemptive on each machine since $\pi$ is non-interruptive. Therefore, it suffices to focus on the starting times of jobs under policy~$\pi_k$ to show that at any time, a subset of available jobs with the highest priority values is scheduled, with ties broken as by $\pi$. Let $y$ be an arbitrary realization of processing times, and let $j_1$ and $j_2$ be jobs where $j_1$ has a higher priority value than $j_2$. We have $S^\pi_{j_1}(y) < S^\pi_{j_2}(y)$ and applying Lemma~\ref{lemma: time mapping} yields $S^{\pi_k}_{j_1} (y) \leq S^{\pi_k}_{j_2}(y)$, completing the argument.
\end{proof}

We introduce the static priority policies \emph{Most Expected Loops first} (MEL), \emph{Least Expected Loops first} (LEL), and \emph{Weighted Least Expected Loops first} (WLEL).
These policies are natural analogs of the established policies for parallel machine scheduling: Longest Expected Processing Time first (LEPT), Shortest Expected Processing Time first (SEPT), and Weighted Shortest Expected Processing Time first (WSEPT), and are defined via the priority values given in Table~\ref{tabel:static policies}.
Combined with Lemma~\ref{lem:priority}, this observation directly yields the following corollary. 

\begin{corollary}
    \label{cor:priority policies}
    MEL, LEL, and WLEL induce the policies LEPT, SEPT, and WSEPT, respectively.
\end{corollary}

\begin{table*}[t] 
\renewcommand{\arraystretch}{1.2}
\centering
\label{tab:priority-policies} \begin{tabular}{| c c c|}
\hline
Flow-shop policy & Auxiliary policy & Priority value of job \(j\)\\ \hline
MEL & LEPT & $\mathbb{E}[Y_j]$ \\ LEL & SEPT & $\frac{1}{\mathbb{E}[Y_j]}$ \\ WLEL & WSEPT & $\frac{w_j}{\mathbb{E}[Y_j]}$\\
\hline
\end{tabular}
\caption{Static priority policies and their counterparts in the auxiliary problems.} 
\label{tabel:static policies}
\end{table*}

We now apply Theorem~\ref{thm:approximation} and Corollary~\ref{cor:priority policies} to derive an optimality result for the MEL and LEL policies under (independent and, in general, non-identical) geometric distributions. Our construction is used to transfer optimality results of \cite{glazebrook1979} and \cite{Weber1979} for the LEPT and SEPT policies for stochastic scheduling on identical parallel machines with integral machine arrivals under geometric distributions and the makespan and total completion time objectives, respectively. Since our auxiliary instances are special cases of this problem, we can apply Theorem~\ref{thm:approximation} with $\alpha=1$.

\begin{corollary}
    \label{corollary:LERL}
    If $Y_1, \ldots, Y_n$ are geometrically distributed, then the MEL policy minimizes the makespan objective in expectation, and the LEL policy minimizes the total completion time objective in expectation for \textsc{Flow Shop Scheduling with Stochastic Reentry}.
\end{corollary}

Note that Corollary~\ref{corollary:LERL} shows that the non-interruptive, static policies MEL and LEL are optimal in the class of interruptive policies for the respective objective.
In fact, the statements of Corollary~\ref{corollary:LERL} can be generalized to a broader class of probability distributions, which \cite{Weber1979} calls monotone hazard rate models, and dynamic, interruptive versions of MEL and LEL.
The definition of this model and the more general statement are relegated to Section~\ref{sec:MHR}.

Since the deterministic special case of \textsc{Flow Shop Scheduling with Stochastic Reentry} is already NP-hard for the total weighted completion time objective \citep{aspern2025}, an efficiently computable and optimal policy cannot exist under this objective in general unless $P = NP$.
However, we can provide an approximation guarantee for the WLEL policy that holds for arbitrary probability distributions and depends only on their squared coefficients of variation.
In particular, for geometrically distributed $Y_j$'s, this gives a $\sqrt 2$-approximation.

\begin{corollary}
    \label{cor:jäger-skutella}
    If $\textup{Var}[Y_j]\leq \Delta \E[Y_j]^2$ for all $j \in [n]$ and some $\Delta \geq 0$, then the WLEL policy for \textsc{Flow Shop Scheduling with Stochastic Reentry} achieves an approximation ratio of $1 + \frac{1}{2}\left( \sqrt{2} - 1 \right) \left ( 1 + \Delta \right )$ in the class of non-interruptive policies for the total weighted completion time objective.
\end{corollary}

Again, this proof follows from Theorem~\ref{thm:approximation} (or rather its restricted-class formulation stated in Proposition~\ref{cor:approximation restricted}) and Corollary~\ref{cor:priority policies} combined with the approximation result for the auxiliary instances in Theorem~\ref{thm: generalized Kawaguchi-Kyan}.
Unlike the optimality results for the auxiliary instances used in Corollary~\ref{corollary:LERL}, this result was not previously known in the literature, and details are established below.

To formulate Theorem~\ref{thm:approximation} for a restricted policy class, let $\Pi$ denote a class of policies for \textsc{Flow Shop Scheduling with Stochastic Reentry} and~$\Pi_k$ a class of policies for the auxiliary problem~$A_k$, $k\in[m]$. We again consider objective functions $\gamma \in \{\sum C_j, \sum w_jC_j, C_{\max}\}$ and let $\textup{OPT}(A_k(I), \Pi_k)$ denote the optimal policy in the class $\Pi_k$.
\begin{proposition}
\label{cor:approximation restricted} 
    If $E[\gamma(\pi_k)] \leq \alpha \E[\gamma(\textup{OPT}(A_k(I), \Pi_k))]$ for all instances $I$ and all $k \in [m]$, and if for all $\pi' \in \Pi$ and all $k \in [m]$ the induced policy satisfies $\pi'_k \in \Pi_k$, then policy $\pi \in \Pi$ is an $\alpha$-approximation for \textsc{Flow Shop Scheduling with Stochastic Reentry} in the policy class $\Pi$.
\end{proposition}

 The proof follows analogously to the proof of Theorem~\ref{thm:approximation},where the following inequality follows, for all $k \in [m]$, from the membership of $\pi^*_k$ in the policy class $\Pi_k$:
    \begin{equation*}
    \E[\gamma(\pi_k)] \leq \alpha \E[\gamma(\textup{OPT}(A_k(I), \Pi_k))] \leq \alpha \E[\gamma(\pi^*_k)].
    \end{equation*}
As observed in Lemma~\ref{lem:priority}, any non-interruptive policy $\pi$ for \textsc{Flow Shop Scheduling with Stochastic Reentry} induces a non-preemptive policy $\pi_k$ for all $k \in [m]$. Thus, the following result completes the proof of Corollary~\ref{cor:jäger-skutella}. Let $X_j$ be an integer-valued random variable for all $j \in [n]$.

\begin{theorem}
\label{thm: generalized Kawaguchi-Kyan}
     If $\textup{Var}[X_j]\leq \Delta \E[X_j]^2$ for all $j \in [n]$ and some $\Delta \geq 0$, then the WSEPT policy for $P | a_i \in \{0,1\}, X_j \sim \textup{stoch} | \sum w_jC_j$ has an approximation ratio of $1 + \frac{1}{2}\left( \sqrt{2} - 1 \right) \left ( 1 + \Delta \right )$ in the class of non-preemptive policies.
\end{theorem}

The proof of Theorem~\ref{thm: generalized Kawaguchi-Kyan} is an extension of a result with identical machine arrival times \citep{jager2018}, combined with ideas of job insertion and splitting \citep{aspern2025}. Technical details are given in Section~\ref{sec:WSEPT}.
The approximation guarantee of $\sqrt{2}$ derived from Corollary~\ref{cor:jäger-skutella} holds for many practically motivated distributions for \textsc{Flow Shop Scheduling with Stochastic Reentry} beyond geometric distributions, see Section~\ref{sec:Binomial}.

\section{Proof of Theorem~\ref{thm: generalized Kawaguchi-Kyan}}
\label{sec:WSEPT}

In this section, we generalize the work of \cite{jager2018} to prove Theorem~\ref{thm: generalized Kawaguchi-Kyan}.
We begin by providing a brief overview of their ideas before presenting our proof.
\cite{jager2018} analyze the \emph{Weighted Shortest Expected Processing Time first} (WSEPT) policy for the stochastic problem $P\lvert X_j \sim \textup{stoch}\rvert\E[\sum w_jC_j]$ and establish a performance guarantee of $1 + \frac{1+\Delta}{2} (1 + \min\{2, \sqrt{2 + 2\Delta}\})^{-1}$, where $\Delta$ is an upper bound on the squared coefficient of variation of the processing times.
They establish this bound in two main steps:
(i) relating the performance guarantee of the stochastic WSEPT policy to the deterministic WSPT policy, and
(ii) improving the guarantee by analyzing $\alpha$-completion times for $\alpha\in\left[0,1\right]$, defined as the time $C_j(\alpha)$ an $\alpha$-fraction of a job $j$'s processing time $p_j$ has been completed.

Our contributions in this section are threefold.
First, we observe that the first step in the analysis of \cite{jager2018} immediately generalizes to the problem with (arbitrary) machine arrivals.
Second, we show in Section~\ref{sec:Kawaguchi-Kyan Bound} that the tight upper bound of $\nicefrac{1}{2}\left( 1 + \sqrt{2}\right)$ for the performance guarantee of the WSPT policy due to \cite{kawaguchi1986} and identical machine arrivals transfers over to the problem with binary machine arrival times and integral processing times.
These two results together yield Theorem~\ref{thm: generalized Kawaguchi-Kyan}.
We then show that the restrictions of the machine arrival and processing times are necessary by providing a counterexample in the general setting in Section~\ref{sec:Counterexamples}.
Third, we show that even under these restrictions, the performance guarantee of WSPT deteriorates for the objective function $\sum w_jC_j(\alpha)$.
Consequently, the second step of the analysis of \cite{jager2018} does not generalize to the problem with binary machine arrivals.

\begin{lemma}
\label{lemma: WSPT-WSEPT}
    If the WSPT policy for $P | a_i| \sum w_jC_j$ is a $(1 + \beta)$-approximation, then the WSEPT policy for $P | a_i, X_j \sim \textup{stoch} | \E[\sum w_j C_j]$ has a performance guarantee of $1 + \beta (1 + \Delta)$ in the class of non-preemptive policies.
\end{lemma}

Because the proof does not depend on the machine environment, it is exactly analogous to the proof contained in \citep{jager2018}, and we shall omit the details here.

\subsection{Generalization of the Kawaguchi-Kyan Bound}
\label{sec:Kawaguchi-Kyan Bound}

We now show that the well-known Kawaguchi-Kyan bound generalizes to the problem setting with binary machine arrival times.
This result replaces the original result by \citet{kawaguchi1986} in the framework of \citet{jager2018} and, together with Lemma~\ref{lemma: WSPT-WSEPT}, completes the proof of Theorem~\ref{thm: generalized Kawaguchi-Kyan}.

\begin{lemma}
\label{lemma: kawaguchi-kyan arrival times}
    The WSPT policy is a $\frac{1}{2}\left( 1 + \sqrt{2}\right)$ approximation for the scheduling problem $P | a_i \in \{0,1\}, p_j \in \mathbb{Z}| \sum w_jC_j$.
\end{lemma}

\begin{proof}
The idea is as follows: Given an instance $I$ of $P | a_i \in \{0,1\}, p_j \in \mathbb{Z}| \sum w_jC_j$, we construct a sequence of instances~$I_h$, $h\in\N$, of $P | | \sum w_jC_j$ such that:
\begin{equation*}
    \lambda(I) \coloneqq \frac{\sum_j w_j C_j^{\textup{WSPT}(I)}}{\sum_j w_j C_j^{\textup{OPT}(I)}} \leq \lim_{h\rightarrow\infty}\frac{\sum_j w_j C_j^{\textup{WSPT}(I_h)}}{\sum_j w_j C_j^{\textup{OPT}(I_h)}}.
\end{equation*}
The statement then follows directly from \citep{kawaguchi1986}.
By applying the same arguments as in \cite[Corollary 1]{Schwiegelshohn2011}, we can assume without loss of generality that $w_j=p_j$ for all jobs~$j$. Furthermore, we assume $n \geq m$; otherwise, consider an instance with the same set of jobs on fewer machines, preserving the ratio.
Let $m_1$ denote the number of machines that arrive at time $1$, $T_{\textup{idle}}$ the first time at which an available machine becomes idle under the WSPT policy, and $n_b$ the number of jobs started before~$T_{\textup{idle}}$ under the WSPT policy. We perform our construction in two main steps.

In the first step, we construct an instance $I^{(2)}$ with $n^{(2)}_b \geq  m^{(2)}_1$ and $\lambda(I) \leq \lambda(I^{(2)})$. These properties enable the second step to be applied. If $T_{\textup{idle}}\geq 2$, then it follows that $n_b \geq m \geq m_1$ and $I^{(2)} = I$ has the desired properties.
If $T_{\textup{idle}} = 1$, then all jobs either start at time $0$ or at time $1$ in the WSPT schedule. Therefore, adding $n-m$ machines with an arrival time of $1$ does not change the cost of the WSPT schedule, but it may decrease the cost of an optimal schedule. Thus, for this augmented instance, which we denote~$I^{(1)}$, we have $n^{(1)} = m^{(1)}$ and $\lambda(I) \leq \lambda(I^{(1)})$. 
If $n^{(1)}_b \geq m^{(1)}_1$, then $I^{(2)} = I^{(1)}$ has the desired properties.
Otherwise, without loss of generality, each machine is assigned exactly one job in both the WSPT schedule and the optimal schedule. 
Since both the WSPT solution and the optimal solution have $m^{(1)}_1$ jobs starting at time $1$ and $n^{(1)}_b$ jobs starting at time 0, there are at least $m^{(1)}_1 - n^{(1)}_b$ jobs that start at time $1$ in both solutions.
Note that these jobs contribute equally to the objective value of both solutions. 
Thus, removing them (along with an equal number of machines with arrival time $1$) can only increase $\lambda$ and gives an instance $I^{(2)}$ with the desired properties.
      
In the second step, the properties of $I^{(2)}$ allow us to construct our sequence of instances~$I_h$.
These are instances of the problem $P | | \sum w_jC_j$, i.e., all machines must be available from time~$0$.
For all $h\in\N$, we first define $I^{(3)}_h$ as follows. 
We set $m^{(3)} = m^{(2)}$ and keep all of the jobs from $I^{(2)}$ for $I^{(3)}_h$. 
For each machine arriving at time $1$ in $I^{(2)}$, we additionally add $h$ dummy jobs of processing time and weight $\nicefrac{1}{h}$ to $I^{(3)}$.
By appropriately inserting these dummy jobs into the WSPT order, this step does not change the completion time of any other jobs in the WSPT schedule.
However, it increases the objective value of the WSPT solution by exactly $\Delta_h \coloneqq  m^{(2)}_1\cdot(\frac{1}{2}+\frac{1}{2h})$ and the objective value of the optimal solution by at most $\Delta_h$.
Similarly to \cite[Theorem 4.17]{aspern2025}, we use job splitting as introduced in \cite[Corollary 3]{Schwiegelshohn2011} to split a total of $m^{(2)}_1$ units of processing that are completed before $T_{\textup{idle}}$ into jobs of processing time and weight equal to $\nicefrac{1}{h}$.
Note that we can do this because $n^{(2)}_b \geq  m^{(2)}_1$, i.e., there exist at least $m_1^{(2)}$ jobs that are processed for at least one time unit before $T_{\textup{idle}}^{(3)} = T_{\textup{idle}}^{(2)}$.
The job splitting decreases the contribution of these jobs by a total of $\Delta'_h\geq  m^{(2)}_1\cdot(\frac{1}{2}-\frac{1}{2h})$.
Hence, we have that: 
\begin{equation*}
    \lim_{h\rightarrow \infty} \Delta'_h \eqqcolon \Delta'\geq \frac{1}{2} m^{(2)}_1 = \Delta \coloneqq \lim_{h\rightarrow \infty} \Delta_h.
\end{equation*}
Again, by appropriately inserting these small jobs into the WSPT order, the contributions of all other jobs remain unchanged under the WSPT policy, while they can only decrease in an optimal solution.

We denote the resulting instance $I_h$ and obtain:
    \begin{align*}
        \lambda(I^{(2)}) & \leq \frac{\sum_j w_j C^{\textup{WSPT}(I^{(2)})}_j + \Delta - \Delta'}{\sum_j w_j C_j^{\textup{OPT}(I^{(2)})}+ \Delta - \Delta'}
        = \lim_{h\rightarrow\infty}\frac{\sum_j w_j C^{\textup{WSPT}(I^{(2)})}_j + \Delta_h - \Delta'_h }{\sum_j w_j C_j^{\textup{OPT}(I^{(2)})} + \Delta_h - \Delta'_h }
        \\
        & \leq \lim_{h\rightarrow\infty}\frac{\sum_j w_j C^{\textup{WSPT}(I_h)}_j}{\sum_j w_j C_j^{\textup{OPT}(I_h)}}.
    \end{align*}

In the first inequality, we use $\Delta'\geq \Delta$, which holds by definition.
The second inequality follows from the fact that both the addition of the small dummy jobs and the job splitting may lead to an improvement in the objective value of the optimal solution, i.e.:
\begin{equation*}
    \sum_j w_j C_j^{\textup{OPT}(I_h)} \leq \sum_j w_j C_j^{\textup{OPT}(I^{(2)})} + \Delta_h - \Delta'_h.
\end{equation*}
The performance guarantee of WSPT on the instances $I_h$ are bounded by the Kawaguchi-Kyan bound, which completes the proof.
\end{proof}

\subsection{Counterexamples}
\label{sec:Counterexamples}

We have proven Lemma~\ref{lemma: kawaguchi-kyan arrival times}, and by extension also Theorem~\ref{thm: generalized Kawaguchi-Kyan}, in a rather restricted setting that suffices for our auxiliary instances.
Namely, we allow machine arrival times to be only~$0$ or~$1$, and we restrict all processing times to be integers.
In the following, we provide a small example to show that neither of the results holds in general when all arrival times are either $0$ or some positive integer $a$, i.e., $a_i \in \{0,a\}$ (or equivalently, if $a_i \in \{0,1\}$ and rational processing times are allowed).

\begin{example}
    Consider the following sequence of instances $I_\kappa$ of $P \vert a_i\in\lbrace 0,a\rbrace \vert \sum w_jC_j$. In instance $I_\kappa$, let there be one job with weight and processing time $\kappa\varphi$, where $\varphi:= \frac{1}{2}\left( 1 + \sqrt{5}\right)$ and $\kappa\in\mathbb{Z}$, and $\kappa$ small jobs with weight and processing time $1$. Let one machine arrive at time~$0$, and $\kappa$ machines at time $\kappa$.

    Any schedule without unforced idleness is a WSPT schedule because $w_j = p_j$ for all jobs~$j$. 
    Hence, we compare the worst-case WSPT schedule (which schedules all small jobs first) with the optimal one (which schedules the large job first).
    As depicted in Figure~\ref{fig:WSPT KK}, this leads to a very large completion time of the large job and no utilization of the additional machines in the worst ordering, whereas the small jobs are evenly distributed between all of the machines with arrival time $\kappa$ in the optimal solution.
    Altering the weights by small terms ensures that the worst-case WSPT schedule is the unique WSPT schedule.
    This yields the following performance guarantee:
    \begin{equation*}
        \begin{split}
            \frac{\sum_j w_j C^{\textup{WSPT}(I_\kappa)}_j}{\sum_j w_j C_j^{\textup{OPT}(I_\kappa)}} = \frac{\kappa^2\varphi\cdot\left(1+\varphi\right) + \frac{\kappa}{2}(\kappa+1)}{\kappa^2\varphi^2 + \kappa\cdot \left(\kappa+1\right)}
            \stackrel{\kappa\rightarrow\infty}{\longrightarrow}\frac{\varphi^2 + \varphi + \frac{1}{2}}{\varphi^2 + 1} = \frac{3+\sqrt{5}}{4}> \frac{1 + \sqrt{2}}{2}.
        \end{split}
    \end{equation*}
    This example motivates the choice of the restriction $a_i \in \{0,1\}$ to establish the $\frac{1}{2} \left ( 1 + \sqrt{2}\right )$- approximation in Lemma~\ref{lemma: kawaguchi-kyan arrival times} further.

      \begin{figure}
      \centering
        \def\ml{1}
        \def\h{0.5}
        \pgfmathsetmacro{\yshift}{\h*(-\ml-5.65)}
        \pgfmathsetmacro{\mone}{\ml + 1}
        \pgfmathsetmacro{\top}{-\h/2}
        \pgfmathsetmacro{\mid}{\h*(-\ml)-\h/2}
        \pgfmathsetmacro{\bottom}{\h*(-\ml)-\h/2}
        \centering
        \begin{subfigure}{0.4\textwidth}
        \centering
    		\begin{tikzpicture}[]		
                
    			\draw[-] (0,0) -- (0, \yshift);
    			\draw[-latex] (0, \yshift) -- node[pos=1,right] {$t$} (4,\yshift);
    
     		    \draw[-,thick] (1,\yshift) -- node[below,yshift=-1mm] {$\kappa$} (1, \yshift-0.2);
                \draw[-,thick] (2.6,\yshift) -- node[below,yshift=-1mm] {$\kappa+\kappa\varphi$} (2.6, \yshift-0.2);

               \foreach \i in {1,...,\ml} { 
                    \node[minimum width=1.6cm, job_l, fill=color2] () at (1,-{\h*\i}) {};
     			}
                
                \foreach \i in {1,2,4,5}{
                    \node[minimum width=1cm, job_l,pattern=north east lines] () at (0,{-\h*(\ml+\i)}) {};
                }

                \foreach \x in {0,0.0666,...,0.999}{
                    \node[minimum width=0.0666cm, job_l, fill=color3] () at (\x,{-\h}) {};
                }

                \node[rotate=90] () at (0.5,{-\h*(\ml+3)}) {$\cdots$};

                \foreach \i in {1,2,3}{
     				\node () at (-0.25,{-\h*(\i)}) {$\i$};
     			}
                \node () at (-0.25,{-\h*(5)}) {$\kappa$};
     			\node () at (-0.55,{-\h*(6)}) {$\kappa+1$};
    	    \end{tikzpicture}
        \caption{An imbalanced WSPT schedule.}
        \end{subfigure}
        \begin{subfigure}{0.4\textwidth}
            	\centering
    	\begin{tikzpicture}[]			
                
    			\draw[-] (0,0) -- (0, \yshift);
    			\draw[-latex] (0, \yshift) -- node[pos=1,right] {$t$} (4,\yshift);
    
     		    \draw[-,thick] (1,\yshift) -- node[below,yshift=-1mm] {$\kappa$} (1, \yshift-0.2);
                \draw[-,thick] (1.6,\yshift) -- node[below,yshift=-1mm] {$\kappa\varphi$} (1.6, \yshift-0.2);

               \foreach \i in {1,...,\ml} { 
                    \node[minimum width=1.6cm, job_l, fill=color2] () at (0,-{\h*\i}) {};
     			}
                
                \foreach \i in {1,2,4,5}{
                    \node[minimum width=1cm, job_l,pattern=north east lines] () at (0,{-\h*(\ml+\i)}) {};
                    \node[minimum width=0.0666cm, job_l, fill=color3] () at (1,{-\h*(\ml+\i)}) {};
                }

                \node[rotate=90] () at (0.5,{-\h*(\ml+3)}) {$\cdots$};

                \foreach \i in {1,2,3}{
     				\node () at (-0.25,{-\h*(\i)}) {$\i$};
     			}
                \node () at (-0.25,{-\h*(5)}) {$\kappa$};
     			\node () at (-0.55,{-\h*(6)}) {$\kappa+1$};
    	\end{tikzpicture}
        \caption{An optimal schedule.}
        \end{subfigure}
        
        \caption{Comparison of a worst-case and an optimal WSPT ordering.}
        \label{fig:WSPT KK}
    \end{figure}
\end{example}

One may conjecture that Lemma~\ref{lemma: kawaguchi-kyan arrival times} can be strengthened to show that the WSPT policy has a performance ratio of $1+\left( 2 \alpha + \sqrt{8\alpha} \right )^{-1}$ for $P | a_i \in \{0,1\}, p_j \in \mathbb{Z}| \sum w_jC_j(\alpha)$ for all $\alpha \in [\nicefrac{1}{2},1]$, as it does for $P | | \sum w_jC_j(\alpha)$.
Here, $C_j(\alpha) = C_j - (1-\alpha)p_j$ denotes the $\alpha$-completion time of job~$j$.
This would lead to a slight strengthening of Theorem~\ref{thm: generalized Kawaguchi-Kyan} as in \citet{jager2018}.
However, this is not the case as the example below demonstrates for $\alpha = \nicefrac{1}{2}$.

\begin{example}
 Consider the following instance of $P \vert a_i\in\lbrace 0,1\rbrace, p_j \in \Z \vert \sum w_jC_j(\frac{1}{2})$. Let there be one job with weight $6$ and processing time $6$, and $5$ small jobs with weight~$1$ and processing time~$1$. Let two machines be available from time $0$ and one machine only from time $1$.

 Again, scheduling the small jobs first yields a worst-case WSPT schedule, while scheduling the large job first yields an optimal one, as illustrated in Figure~\ref{fig:alpha=1/2}. It can be computed that:

     \begin{equation*}
        \begin{split}
        \frac{\sum_j w_j C^{\textup{WSPT}}_j(\frac{1}{2})}{\sum_j w_j C_j^{\textup{OPT}}(\frac{1}{2})} = \frac{35.5}{26.5} 
        > 4/3 = 1 + \frac{1}{2 \cdot \frac{1}{2} + \sqrt{8 \cdot \frac{1}{2}}}.
        \end{split}
    \end{equation*}
While it is not clear that the approximation ratio of $1 + \frac{1}{2}\left( \sqrt{2} - 1 \right) \left ( 1 + \Delta \right )$ is tight for the WSEPT policy, the above example shows that it is not possible to obtain a stronger bound by using the framework of \citet{jager2018}.

\begin{figure}
    	\centering
        \definecolor{color0}{rgb}{0.95, 0.95, 0.95}
        \begin{subfigure}{0.4\textwidth}
        \centering
    		\begin{tikzpicture}[]			
                
    			\draw[-] (0,0) -- (0, -2);
    			\draw[-latex] (0, -1.85) -- node[pos=1,right] {$t$} (4.5,-1.85);
    
     		    \draw[-,thick] (1,-1.85) -- node[below,yshift=-1mm] {$2$} (1, -2.05);
                \draw[-,thick] (4,-1.85) -- node[below,yshift=-1mm] {$8$} (4, -2.05);

               \foreach \i in {1,...,3} {
                    \foreach \x in {0,1} {
                    \node[minimum width=0.5cm, job_l, fill=color3] () at (\x*0.5,-{\i*0.5}) { };
                    }
     			}
                \node[minimum width=0.5cm, job_l, fill=white] () at (0,-1.5) { };
                \node[minimum width=0.5cm, job_l, pattern=north east lines] () at (0,-1.5) { };
                \node[minimum width=3cm, job_l, fill=color2] () at (1,-0.5) { };

    	   \end{tikzpicture}
           \caption{A WSPT solution.}
        \end{subfigure}
        \begin{subfigure}{0.4\textwidth}
        \centering
    		\begin{tikzpicture}[]			
                
    			\draw[-] (0,0) -- (0, -2);
    			\draw[-latex] (0, -1.85) -- node[pos=1,right] {$t$} (4.5,-1.85);
    
     		    \draw[-,thick] (1.5,-1.85) -- node[below,yshift=-1mm] {$3$} (1.5, -2.05);
                \draw[-,thick] (3,-1.85) -- node[below,yshift=-1mm] {$6$} (3, -2.05);

               \foreach \i in {2,3} {
                    \foreach \x in {0,1,2} {
                    \node[minimum width=0.5cm, job_l, fill=color3] () at (\x*0.5,-{\i*0.5}) { };
                    }
     			}
                \node[minimum width=0.5cm, job_l, fill=white] () at (0,-1.5) { };
                \node[minimum width=0.5cm, job_l, pattern=north east lines] () at (0,-1.5) { };
                \node[minimum width=3cm, job_l, fill=color2] () at (0,-0.5) { };

    	   \end{tikzpicture}
           \caption{An optimal solution.}
        \end{subfigure}
        \caption{A counterexample for $\alpha = 1/2$.}
        \label{fig:alpha=1/2}
\end{figure}
\end{example}

\section{Extensions}
\label{sec:Extensions}

\subsection{Monotone Hazard Rate Model}
\label{sec:MHR}

In this section, we present the monotone hazard rate model introduced by \citet{Weber1979} and state a more general version of Corollary~\ref{corollary:LERL}. Let $h_j: \mathbb{N}_0 \rightarrow [0,1]$ denote the probability that job~$j$ is completed by loop $l+1$, conditioned on it not being completed after loop $l$.

We say that $\{Y_1, Y_2, \ldots, Y_n\}$ have a \emph{monotone hazard rate} if the following hold:

\begin{enumerate}
    \item There exist $\nu$ functions $f_1, \ldots, f_\nu: \mathbb{N}_0 \rightarrow [0,1]$ that are either all non-increasing or all non-decreasing such that:
    $$f_c(l_1) < f_{c+1}(l_2) \ \textup{ for all } c \in [\nu-1], l_1, l_2 \in \mathbb{N}_0.$$
    \item For all $j\in [n]$, there exists an $c \in [\nu]$ and $x_j \in \mathbb{N}_0$ such that:
    $$h_j(l) = f_c(x_j+l) \ \textup{ for all } l \in \mathbb{N}_0.$$
\end{enumerate}

\citet{Weber1979} shows that the preemptive LEPT and SEPT policies minimize the expected makespan and expected total completion time, respectively, for \textsc{Stochastic Scheduling with Machine Arrivals} if $\{Y_1, Y_2, \ldots, Y_n\}$ have a monotone hazard rate. 

Inspired by this result, we introduce the dynamic and interruptive priority policies \emph{Most Expected Remaining Loops first} (MERL) and \emph{Least Expected Remaining Loops first} (LERL).
These policies are defined by the same priority values as their static counterparts from Section~\ref{sec:LERL&MERL}, but the priority value is updated each time a loop is completed.
In particular, if job~$j$ re-enters the flow shop after completing loop~$l$, its priority value is $\E[Y_j - l \vert Y_j > l]$ under the MERL policy, or $\E[Y_j - l \vert Y_j > l]^{-1}$ under the LERL policy.
Analogously to Lemma~\ref{lem:priority} and Corollary~\ref{cor:priority policies}, it can be shown that the MERL and LERL policies induce the preemptive LEPT and SEPT policies, respectively.
Here, we apply the arguments of Lemma~\ref{lem:priority} to the starting time of each individual loop rather than focusing only on the starting time of the first loop.
Combined with Theorem~\ref{thm:approximation}, this gives the following result.

\begin{corollary}
    If $\{Y_1, Y_2, \ldots, Y_n\}$ have a monotone hazard rate, then the MERL policy minimizes the expected makespan, and the LERL policy minimizes the expected total completion time for \textsc{Flow Shop Scheduling with Stochastic Reentry}.
    \label{corollary:LERL2}
\end{corollary}

The monotone hazard rate model generalizes several different models \citep{Weber1979}. 
The following two are of particular interest because they capture the results in Corollary~\ref{corollary:LERL} and in \citep{yu2020} and \citep{aspern2025}.
If all $Y_j$ are geometrically distributed with parameters $q_j$, then we have that $h_j(l) = q_j$ for all $l \in \mathbb{N}_0$.
Appropriately choosing constant functions $f_1, \ldots, f_\nu$ and setting $x_j = 0$ for all $j$ shows that $\{Y_1, Y_2, \ldots, Y_n\}$ have a monotone hazard rate.
By the memorylessness property of the geometric distribution, MERL and LERL are equivalent to the non-interruptive policies MEL and LEL in this case.
Thus, Corollary~\ref{corollary:LERL2} is a strict generalization of Corollary~\ref{corollary:LERL}.
If all $Y_j$ are deterministic, then we can choose a single function $f$ by defining $f(l) = 1$ for $l \geq Y_{\max}$ and $f(l) = 0$ for $l < Y_{\max}$, where $Y_{\max} \coloneqq \max_j Y_j$.
Now we can set $x_j = Y_{\max} - Y_j$ and we have $h_j(l) = f(x_j + l)$.
In this case, MERL and LERL are equivalent to the deterministic scheduling rules \emph{Most Remaining Loops first} (MRL) and \emph{Least Remaining Loops first} (LRL), which were previously studied by \citet{yu2020} and \citet{aspern2025}.

\subsection{Practically Motivated Distributions}
\label{sec:Binomial}
In this section, we present several examples involving models with different practically motivated distributions, in which we can apply Corollary~\ref{cor:jäger-skutella} to show that WLEL has an approximation guarantee of $\sqrt{2}$ for the total weighted completion time objective, underscoring the strength of this theoretical result. Afterward, we show that the optimality results derived in Corollary~\ref{corollary:LERL2} for the makespan objective cannot be extended to these generalized models.

Consider a model in which each job~$j$ has a known number of required loops, $\mathcal L_j$, and each of its loops has a probability of success $q_j$. If a loop fails (e.g., due to a processing error), it must be completed again before the next loop of this job can be processed. This can be captured by negative binomial distributions, which generalize the deterministic and geometric cases.

As another example, consider a model similar to the previous one, but now the entire job must be restarted from the beginning if a loop fails. Thus, the completion time of job~$j$ is the first time of $\mathcal L_j$ consecutive successful loops. 

For both of these families of distributions, one can derive an upper bound of $\Delta\leq 1$, e.g., using results from~\cite{Philippou1983}. In general, this property can be shown for all distributions in a model, where the number of loops~$Y_j$ for job~$j$ is, at any point in time, at most the initial expected number, i.e.: 
    \begin{equation*}
        \E\left[Y_j - x \vert Y_j >x\right] \leq \E\left[Y_j\right] \ \textup{ for all } x\geq0.
    \end{equation*}
For these so-called NBUE distributions, \citet{hall1981} proved $\Delta\leq 1$, and, therefore, Corollary~\ref{cor:jäger-skutella} implies that the WLEL policy is a $\sqrt 2$-approximation.

In the remaining part of this section, we show that the MERL policy is no longer optimal for the makespan objective when both deterministic and geometrically distributed jobs are allowed. Thus, the optimality results derived in Corollary~\ref{corollary:LERL2} cannot be generalized in this direction.
Note that in this case, the MERL policy is no longer non-interruptive because the (expected) number of remaining loops of deterministic jobs changes over time.
\begin{example}
Consider an instance with $m=2$ machines, two deterministic jobs that require two loops for completion, i.e., $\mathcal L_1 = \mathcal L_2 = 2$, and one stochastic job with $Y_3$ geometrically distributed with success probability~$q_3 = \nicefrac{1}{2}+\varepsilon$ (i.e., $\E[Y_3] < \mathcal L_1 = \mathcal L_2$).

\begin{figure*}
	\centering
    \begin{subfigure}{0.4\textwidth}
    \begin{tikzpicture}[scale=0.5]
			
			\pgfmathsetmacro{\m}{2.65}
			
			\draw[-] (0,0) -- (0, -\m);
			\draw[-latex] (0, -\m) -- node[pos=1,right] {$t$} (10.5, -\m);

			
			\foreach \x in {1,3}{
				\node[job1] (1) at (\x,-1) {}; \node at (1.center) [] {$1$};
				\node[job1] (1) at (\x+1,-2) {}; \node at (1.center) [] {$1$};
			}
			\foreach \x in {2,4}{
				\node[job2] (2) at (\x,-1) {}; \node at (2.center) [] {$2$};
				\node[job2] (2) at (\x+1,-2) {}; \node at (2.center) [] {$2$};
			}
			\foreach \x in {0}{
				\node[job3] (3) at (\x,-1) {}; \node at (3.center) [] {$3$};
				\node[job3] (3) at (\x+1,-2) {}; \node at (3.center) [] {$3$};
			}
            \draw[-,thick] (6,-0.25-\m) -- node[below,yshift=-1mm] {$C_{\max}=6$} (6, 0.25-\m);
	\end{tikzpicture}
    \caption{LV for $Y_3 = 1$}	
    \end{subfigure}
    \begin{subfigure}{0.4\textwidth}
    \begin{tikzpicture}[scale=0.5]
			
			\pgfmathsetmacro{\m}{2.65}
			
			\draw[-] (0,0) -- (0, -\m);
			\draw[-latex] (0, -\m) -- node[pos=1,right] {$t$} (10.5, -\m);

			
			\foreach \x in {0,3}{
				\node[job1] (1) at (\x,-1) {}; \node at (1.center) [] {$1$};
				\node[job1] (1) at (\x+1,-2) {}; \node at (1.center) [] {$1$};
			}
			\foreach \x in {1,4}{
				\node[job2] (2) at (\x,-1) {}; \node at (2.center) [] {$2$};
				\node[job2] (2) at (\x+1,-2) {}; \node at (2.center) [] {$2$};
			}
			\foreach \x in {2}{
				\node[job3] (3) at (\x,-1) {}; \node at (3.center) [] {$3$};
				\node[job3] (3) at (\x+1,-2) {}; \node at (3.center) [] {$3$};
			}
            \draw[-,thick] (6,-0.25-\m) -- node[below,yshift=-1mm] {$C_{\max}=6$} (6, 0.25-\m);
	\end{tikzpicture}
    \caption{MERL for $Y_3 = 1$}	
    \end{subfigure}
    
    \begin{subfigure}{0.4\textwidth}
    \begin{tikzpicture}[scale=0.5]
			
			\pgfmathsetmacro{\m}{2.65}
			
			\draw[-] (0,0) -- (0, -\m);
			\draw[-latex] (0, -\m) -- node[pos=1,right] {$t$} (10.5, -\m);

			
			\foreach \x in {1,4}{
				\node[job1] (1) at (\x,-1) {}; \node at (1.center) [] {$1$};
				\node[job1] (1) at (\x+1,-2) {}; \node at (1.center) [] {$1$};
			}
			\foreach \x in {3,5}{
				\node[job2] (2) at (\x,-1) {}; \node at (2.center) [] {$2$};
				\node[job2] (2) at (\x+1,-2) {}; \node at (2.center) [] {$2$};
			}
			\foreach \x in {0,2}{
				\node[job3] (3) at (\x,-1) {}; \node at (3.center) [] {$3$};
				\node[job3] (3) at (\x+1,-2) {}; \node at (3.center) [] {$3$};
			}
            \draw[-,thick] (7,-0.25-\m) -- node[below,yshift=-1mm] {$C_{\max}=7$} (7, 0.25-\m);
	\end{tikzpicture}
    \caption{LV for $Y_3 = 2$}	
    \end{subfigure}
    \begin{subfigure}{0.4\textwidth}
    \begin{tikzpicture}[scale=0.5]
			
			\pgfmathsetmacro{\m}{2.65}
			
			\draw[-] (0,0) -- (0, -\m);
			\draw[-latex] (0, -\m) -- node[pos=1,right] {$t$} (10.5, -\m);

			
			\foreach \x in {0,3}{
				\node[job1] (1) at (\x,-1) {}; \node at (1.center) [] {$1$};
				\node[job1] (1) at (\x+1,-2) {}; \node at (1.center) [] {$1$};
			}
			\foreach \x in {1,5}{
				\node[job2] (2) at (\x,-1) {}; \node at (2.center) [] {$2$};
				\node[job2] (2) at (\x+1,-2) {}; \node at (2.center) [] {$2$};
			}
			\foreach \x in {2,4}{
				\node[job3] (3) at (\x,-1) {}; \node at (3.center) [] {$3$};
				\node[job3] (3) at (\x+1,-2) {}; \node at (3.center) [] {$3$};
			}
            \draw[-,thick] (7,-0.25-\m) -- node[below,yshift=-1mm] {$C_{\max}=7$} (7, 0.25-\m);
	\end{tikzpicture}
    \caption{MERL for $Y_3 = 2$}	
    \end{subfigure}
    
    \begin{subfigure}{0.4\textwidth}
    \begin{tikzpicture}[scale=0.5]
			
			\pgfmathsetmacro{\m}{2.65}
			
			\draw[-] (0,0) -- (0, -\m);
			\draw[-latex] (0, -\m) -- node[pos=1,right] {$t$} (10.5, -\m);

			
			\foreach \x in {1,5}{
				\node[job1] (1) at (\x,-1) {}; \node at (1.center) [] {$1$};
				\node[job1] (1) at (\x+1,-2) {}; \node at (1.center) [] {$1$};
			}
			\foreach \x in {3,6}{
				\node[job2] (2) at (\x,-1) {}; \node at (2.center) [] {$2$};
				\node[job2] (2) at (\x+1,-2) {}; \node at (2.center) [] {$2$};
			}
			\foreach \x in {0,2,4}{
				\node[job3] (3) at (\x,-1) {}; \node at (3.center) [] {$3$};
				\node[job3] (3) at (\x+1,-2) {}; \node at (3.center) [] {$3$};
			}
            \draw[-,thick] (8,-0.25-\m) -- node[below,yshift=-1mm] {$C_{\max}=8$} (8, 0.25-\m);
	\end{tikzpicture}
    \caption{LV for $Y_3 = 3$}	
    \end{subfigure}
    \begin{subfigure}{0.4\textwidth}
    \begin{tikzpicture}[scale=0.5]
			
			\pgfmathsetmacro{\m}{2.65}
			
			\draw[-] (0,0) -- (0, -\m);
			\draw[-latex] (0, -\m) -- node[pos=1,right] {$t$} (10.5, -\m);

			
			\foreach \x in {0,3}{
				\node[job1] (1) at (\x,-1) {}; \node at (1.center) [] {$1$};
				\node[job1] (1) at (\x+1,-2) {}; \node at (1.center) [] {$1$};
			}
			\foreach \x in {1,5}{
				\node[job2] (2) at (\x,-1) {}; \node at (2.center) [] {$2$};
				\node[job2] (2) at (\x+1,-2) {}; \node at (2.center) [] {$2$};
			}
			\foreach \x in {2,4,6}{
				\node[job3] (3) at (\x,-1) {}; \node at (3.center) [] {$3$};
				\node[job3] (3) at (\x+1,-2) {}; \node at (3.center) [] {$3$};
			}
            \draw[-,thick] (8,-0.25-\m) -- node[below,yshift=-1mm] {$C_{\max}=8$} (8, 0.25-\m);
	\end{tikzpicture}
    \caption{MERL for $Y_3 = 3$}	
    \end{subfigure}
    
    \begin{subfigure}{0.4\textwidth}
    \begin{tikzpicture}[scale=0.5]
			
			\pgfmathsetmacro{\m}{2.65}
			
			\draw[-] (0,0) -- (0, -\m);
			\draw[-latex] (0, -\m) -- node[pos=1,right] {$t$} (10.5, -\m);

			
			\foreach \x in {1,5}{
				\node[job1] (1) at (\x,-1) {}; \node at (1.center) [] {$1$};
				\node[job1] (1) at (\x+1,-2) {}; \node at (1.center) [] {$1$};
			}
			\foreach \x in {3,7}{
				\node[job2] (2) at (\x,-1) {}; \node at (2.center) [] {$2$};
				\node[job2] (2) at (\x+1,-2) {}; \node at (2.center) [] {$2$};
			}
			\foreach \x in {0,2,4,6}{
				\node[job3] (3) at (\x,-1) {}; \node at (3.center) [] {$3$};
				\node[job3] (3) at (\x+1,-2) {}; \node at (3.center) [] {$3$};
			}
            \draw[-,thick] (9,-0.25-\m) -- node[below,yshift=-1mm] {$C_{\max}=9$} (9, 0.25-\m);
	\end{tikzpicture}
    \caption{LV for $Y_3 = 4$}	
    \end{subfigure}
    \begin{subfigure}{0.4\textwidth}
    \begin{tikzpicture}[scale=0.5]
			
			\pgfmathsetmacro{\m}{2.65}
			
			\draw[-] (0,0) -- (0, -\m);
			\draw[-latex] (0, -\m) -- node[pos=1,right] {$t$} (10.5, -\m);

			
			\foreach \x in {0,3}{
				\node[job1] (1) at (\x,-1) {}; \node at (1.center) [] {$1$};
				\node[job1] (1) at (\x+1,-2) {}; \node at (1.center) [] {$1$};
			}
			\foreach \x in {1,5}{
				\node[job2] (2) at (\x,-1) {}; \node at (2.center) [] {$2$};
				\node[job2] (2) at (\x+1,-2) {}; \node at (2.center) [] {$2$};
			}
			\foreach \x in {2,4,6,8}{
				\node[job3] (3) at (\x,-1) {}; \node at (3.center) [] {$3$};
				\node[job3] (3) at (\x+1,-2) {}; \node at (3.center) [] {$3$};
			}
            \draw[-,thick] (10,-0.25-\m) -- node[below,yshift=-1mm] {$C_{\max}=10$} (10, 0.25-\m);
	\end{tikzpicture}
    \caption{MERL for $Y_3 = 4$}	
    \end{subfigure}

    \caption{LV policy vs. MERL policy under different realizations of $Y_3$.}
    \label{fig:LV-MERL}
\end{figure*}
One can observe that the \emph{Largest Variance first} (LV) policy with tie breaking according to MERL (a non-preemptive version of this policy appears in \cite{pinedo1987}) has the same objective value as the MERL policy for $Y_3 \leq 3$ and a strictly smaller objective value for~$Y_3 \geq 4$. 
Thus, the LV policy has a strictly smaller expected makespan than the MERL policy.
The resulting schedules of both policies for $Y_3 \leq 4$ are depicted in Figure~\ref{fig:LV-MERL}.
\end{example}

\section{Discussion}
\label{sec:outlook}

In this paper, we establish the first optimality results and approximation guarantees for \textsc{Flow Shop Scheduling with Stochastic Reentry}.
The main tool for establishing these results is an aggregate completion-time-preserving reduction to stochastic parallel-machine scheduling with binary machine arrivals.
This reduction enables extensions and, in some cases, direct transfers of classical optimality results and approximation guarantees from the literature.
At the same time, the established connection provides insights into the structure of the reentrant flow shop that have inspired the development of our approximate policies.
The reduction provides a basis for designing and analyzing further policies---whether optimal, approximate, or heuristic---for the reentrant flow shop model under various stochastic models that capture real-world uncertainty in production environments.

In Section~\ref{sec:Binomial}, we introduce a family of practically motivated probability distributions that we expect to be particularly relevant for future research. Our result in Corollary~\ref{cor:jäger-skutella} immediately implies that the WLEL policy is a $\sqrt 2$-approximation for the total weighted completion time objective under these distributions.
While we show in Section~\ref{sec:Binomial} that MERL is no longer optimal under these distributions, it would be interesting to establish worst-case performance guarantees in general, or in settings where either the number of loops or the success probability is constant across all jobs, or to develop new policies specifically tailored to these distributions.

Finally, it is important to note that our reduction heavily depends on the assumption of unit processing times. 
For future research, we propose establishing structural results on flow shops with reentry---both deterministic and stochastic---where this assumption is relaxed.
Such results would allow us to study more complex manufacturing settings, such as reentrant flow shops with bottleneck machines.

\section*{Acknowledgments}
Maximilian von Aspern and Felix Buld were funded by the German Research Foundation (DFG) under project number 277991500. Parts of this work were carried out during a research visit to Columbia University in the City of New York. Felix Buld was supported by a scholarship from the German Academic Exchange Service (DAAD).

\bibliography{bibliography}

\end{document}